\documentclass[conference]{IEEEtran}
\IEEEoverridecommandlockouts
\usepackage[T1]{fontenc} 
\usepackage{cite}
\usepackage{graphicx}
\usepackage{epstopdf}
\usepackage{amsmath,amssymb,amsfonts}
\usepackage{amsthm}
\usepackage{algorithm}
\usepackage{algorithmic}
\usepackage{textcomp}
\usepackage{xcolor}
\usepackage{graphicx}
\usepackage{tikz}
\usetikzlibrary{arrows.meta, positioning, decorations.markings}
\usepackage{tabularx,booktabs}
\usepackage[flushleft]{threeparttable}

\usepackage{booktabs,ragged2e}
\usepackage{multirow}
\usepackage{stfloats}
\usepackage{placeins} 
\usepackage{graphicx}
\usepackage{subcaption}

\usepackage{bm}
\usepackage[colorlinks,
            linkcolor=blue,
            anchorcolor=blue,
            citecolor=blue
            ]{hyperref}
\hypersetup{draft}

\newtheorem{proposition}{Proposition}

\newtheorem{remark}{Remark}

\def\BibTeX{{\rm B\kern-.05em{\sc i\kern-.025em b}\kern-.08em
    T\kern-.1667em\lower.7ex\hbox{E}\kern-.125emX}}

\begin{document}

\title{Enabling Training-Free Semantic Communication Systems with Generative Diffusion Models 

}

\author{
\thanks{This work is partly supported by the National Key R\&D Program of China under Grant No. 2024YFE0200802, partly by NSFC under grant No.62293481 and No.62201505, and partly supported by the China Scholarship Council (No. 202406320381) and the CAST Young Talent Support Program for Doctoral Students.}
\thanks{S. Tang and Y. Jia contributed equally to this work.}
\IEEEauthorblockN{
Shunpu Tang\IEEEauthorrefmark{2}, 
Yuanyuan Jia\IEEEauthorrefmark{2},  
Qianqian Yang\IEEEauthorrefmark{2}\IEEEauthorrefmark{1},
Ruichen Zhang\IEEEauthorrefmark{4},
Jihong Park\IEEEauthorrefmark{6},
Dusit Niyato\IEEEauthorrefmark{4}
    }

\IEEEauthorblockA{ 
\IEEEauthorrefmark{2}College of Information Science and Electronic Engineering, Zhejiang University, Hangzhou, China \\
\IEEEauthorrefmark{4}College of Computing and Data Science, Nanyang Technological University, Singapore \\
\IEEEauthorrefmark{6} ISTD Pillar, Singapore University of Technology and Design, Singapore  \\
Email: \{tangshunpu, labulado, qianqianyang20\}@zju.edu.cn, 
\{ruichen.zhang, dniyato\}@ntu.edu.sg, jihong\_park@sutd.edu.sg
}
}

\maketitle

\begin{abstract}
Semantic communication (SemCom) has recently emerged as a promising paradigm for next-generation wireless systems. Empowered by advanced artificial intelligence (AI) technologies, SemCom has achieved significant improvements in transmission quality and efficiency. However, existing SemCom systems either rely on training over large datasets and specific channel conditions or suffer from performance degradation under channel noise when operating in a training-free manner. To address these issues, we explore the use of generative diffusion models (GDMs) as training-free SemCom systems. Specifically, we design a semantic encoding and decoding method based on the inversion and sampling process of the denoising diffusion implicit model (DDIM), which introduces a two-stage forward diffusion process, split between the transmitter and receiver to enhance robustness against channel noise. Moreover, we optimize sampling steps to compensate for the increased noise level caused by channel noise. We also conduct a brief analysis to provide insights about this design. Simulations on the Kodak dataset validate that the proposed system outperforms the existing baseline SemCom systems across various metrics. 
    \end{abstract}

\section{Introduction}
Semantic communication (SemCom) has emerged as a promising paradigm for the next generation of wireless communication systems and has attracted significant research interest in recent years. The key idea of SemCom is to understand the meaning of the transmitted data and transmit the most relevant information to the receiver, which can significantly reduce the amount of data transmitted over the channel and support downstream tasks at the receiver \cite{Semantic1}, such as autonomous driving, metaverse, and smart cities. 

Benefiting from recent advances in artificial intelligence (AI) technologies, various SemCom systems have been proposed. Specifically, the authors in \cite{deepjscc} proposed a deep joint source-channel coding (DeepJSCC) system for wireless image transmission, where the semantic encoder and decoder are implemented using convolutional neural networks (CNNs) to extract the semantic information behind the original pixels and reconstruct the original image.  DeepJSCC achieves superior reconstruction performance compared to conventional digital communication systems and effectively mitigates the cliff effect. Following this work, the authors in \cite{SWINJSCC} modified the architecture of the semantic encoder and decoder by introducing the powerful transformer backbone, significantly improving transmission quality.  In addition, a contrastive learning-based SemCom Framework~\cite {Shunpu_SemCom} was proposed to train the semantic encoder and decoder, which enhances semantic consistency during transmission. 
\textcolor{black}{However, due to the use of the autoencoder architectures and the discriminative AI paradigm, these approaches struggle to achieve high communication efficiency, and also require extensive training on large datasets and various channel conditions, which significantly limit their performance and flexibility in practical systems.}

\textcolor{black}{Fortunately, the recent emergence of generative artificial intelligence (GenAI) offers new opportunities to overcome this limitation. By learning to capture the underlying data distribution rather than direct input-to-label mappings in discriminative AI~\cite{AIGC}, GenAI enables the generation of high-dimensional data (e.g., images or text) from low-dimensional vectors. This allows SemCom systems to transmit minimal data while enabling the receiver to reconstruct the original content through conditional sampling from the learned distribution~\cite{10614204}.} \textcolor{black}{For example, the authors in~\cite{GAN_JSCC} proposed a generative JSCC framework in which a generative adversarial network (GAN) is integrated into the decoder to enhance reconstruction quality by leveraging the semantic priors learned by the generator. Moreover, more powerful generative diffusion models (GDMs)\cite{ho2020denoising, song2021denoising} have been introduced into DeepJSCC-based frameworks\cite{Yilmaz-Other-2024, Chen-arXiveessIV-2025}, where the degraded images reconstructed by DeepJSCC serve as conditional inputs to guide the diffusion sampling process, resulting in more realistic and perceptually faithful reconstructions. However, these approaches remain dependent on a well-trained DeepJSCC system.}

\textcolor{black}{To eliminate the dependency on training, recent studies have explored training-free generative SemCom frameworks. Specifically, the authors in~\cite{tang2024evolving} introduced a channel-aware GAN inversion method for semantic encoding and employed the same GAN generator for decoding, thus avoiding any encoder–decoder training on the communication task. Meanwhile, the works in~\cite{Cicchetti_MLSP-2024,tang2024retrieval, Li_SemCom} proposed to use a pretrained image captioner to extract textual descriptions from source images. These captions, along with auxiliary visual features such as edge maps or latent vectors, are transmitted to the receiver, where the caption serves as a conditioning prompt to guide a GDM in reconstructing the original image. However, these training-free approaches lack robustness to channel noise, as the transmitted semantic conditions, such as captions or latent features, are easily corrupted during transmission, leading to degraded reconstruction quality.
}

To overcome these limitations,  we propose a fully training-free generative SemCom framework that leverages publicly available pretrained GDMs. Specifically, we design a semantic encoding and decoding method based on the sampling and inversion processes of the denoising diffusion implicit model (DDIM) \cite{song2021denoising}, which introduces a two-stage forward diffusion strategy split between the transmitter and receiver to enhance robustness against channel noise.
In addition, we optimize the number of sampling steps at the receiver to match the total noise level introduced by the channel noise. We further provide a brief analysis of this design, how channel noise shifts the latent distribution away from the latent distribution of GDM in ideal conditions, and give insights into how the proposed system mitigates this mismatch. Simulations on the Kodak dataset are conducted to demonstrate the superiority of the proposed system across a wide range of perceptual and distortion metrics. In particular, our method achieves over 50\% performance gains in Fréchet Inception distance (FID) compared to baselines when the SNR is below 5 dB. 
\section{Preliminaries}
\subsection{System model of SemCom}
In this paper, we consider a typical SemCom system for wireless image transmission, where the transmitter and receiver are equipped with a semantic encoder and decoder, respectively. For the input RGB image $\bm{x} \in \mathbb{R}^{3 \times H \times W}$, \textcolor{black}{where $H$ and $W$ denote the image height and width, respectively}, the semantic encoder first extracts the semantic information and directly maps it into to a $k$ complex-dimensional channel input signal $\bm{z} \in \mathbb{C}^{k}$, given by
\begin{equation}
    \bm{z} = \mathcal{E}_{\theta}(\bm{x}),
\end{equation}
where $\mathcal{E}_{\theta}(\cdot)$ is the semantic encoder with parameters $\theta$. To evaluate the communication efficiency, we define the bandwidth compression ratio as $\text{BCR} = \frac{k}{N}$, where $N=3 \times H \times W$ denotes the source bandwidth. Then, the channel input $\bm{z}$ is transmitted over a noisy channel, which is modeled as
\begin{equation}
    \bm{y} = \bm{z} + \bm{n},
\end{equation}
where $\bm{n} \sim \mathcal{CN}(0, \sigma_{\text{ch}}^2\bm{I})$ is the additive white Gaussian noise (AWGN) with zero mean and variance $\sigma_{\text{ch}}^2$. At the receiver side, the semantic decoder reconstructs the image $\hat{\bm{x}}$ from the received signal $\bm{y}$, which can be expressed as
\begin{equation}
    \hat{\bm{x}} = \mathcal{D}_{\phi}(\bm{y}),
\end{equation}
where $\mathcal{D}_{\phi}(\cdot)$ represents the semantic decoder with parameter $\phi$. The performance of the SemCom system can be assessed by the difference between $\bm{x}$ and $\hat{\bm{x}}$ using various metrics, including distortion metrics such as peak signal-to-noise ratio (PSNR) and multi-scale structural similarity (MS-SSIM), human perceptual metrics like learned perceptual image patch similarity (LPIPS), and distribution metrics such as Fréchet Inception distance (FID).

\subsection{Generative Diffusion Models}
GDMs are a class of generative models that learns to generate data by gradually denoising from a pure noise distribution. Specifically, a typical diffusion model consists of a forward diffusion process with no learnable parameters and a reverse denoising process with a learnable neural network\cite{ho2020denoising}. Exemplifying the image generation task with latent diffusion, the forward diffusion process gradually adds Gaussian noise to the training data $\bm{z}$, which can be expressed as
\begin{equation}
    q(\bm{z}_t | \bm{z}_{t-1}) = \mathcal{N}( \sqrt{\alpha_t}\bm{z}_{t-1}, (1-\alpha_t)\bm{I}),
\end{equation}
where $\bm{x}_t$ is the noisy image at time step $t$, and \textcolor{black}{$\alpha_t\in (0,1)$ is a hyperparameter} that schedules the noise level. Therefore, given a training image $\bm{z}_0$, after ${T_{F}}$ steps of forward diffusion, we can directly write the final noisy latent $\bm{z}_{T_{F}}$ through the reparameterization trick as
\begin{equation}
\label{eq:ddpm_forward}
    \bm{z}_{T_{F}} \sim  \mathcal{N}( \sqrt{\bar{\alpha}_{T_{F}}}\bm{z}_0,(1-\bar{\alpha}_{T_{F}})\bm{I}),
\end{equation}
where $\overline{\alpha}_{T_{F}} = \prod_{i=1}^{T_{F}}  \alpha_i$ is the cumulative product of the noise schedule terms, and $\bm{\epsilon} \sim \mathcal{N}(0, \bm{I})$ is a Gaussian noise. \textcolor{black}{We note that $\alpha_0>\alpha_t>\cdots > \alpha_{T_F}$ is satisfied in the training process to  make sure that $\bar{\alpha}_{T_{F}}$ monotonically decreases as $T_F$ increases.}

For the reverse denoising process, the model learns to iteratively denoise the image by predicting the noise added to the image at each time step, given by
\begin{equation}
    p_\omega(\bm{z}_{t-1} | \bm{z}_t) = \mathcal{N}( \mu_\omega(\bm{z}_t, t), \Sigma_\omega(\bm{z}_t, t)),
    \label{eq:reverse_denoising}
\end{equation}
where $\Sigma_\omega(\bm{z}_t, t)$ is a predefined variance and $\mu_\omega(\bm{z}_t, t)$ is the predicted mean, given as
\begin{equation}
    \mu_\omega(\bm{z}_t, t) = \frac{1}{\sqrt{\alpha_t}}(\bm{z}_t - \frac{1-\alpha_t}{\sqrt{1-\overline{\alpha}_t}}\epsilon_\omega(\bm{z}_t, t)).
\end{equation}
The denoising term $\epsilon_\omega$ is typically modeled as a neural network with U-Net architecture to predict the noise $\epsilon_\omega(\bm{z}_t, t)$ to be subtracted from the image at each time step. This neural network is trained using the following loss function:
\begin{equation}
    \mathcal{L}(\omega) = \mathbb{E}_{\bm{z}_0, \bm{\epsilon}, t} \left[ ||\bm{\epsilon} - \epsilon_\omega(\bm{z}_t, t)||^2 \right],
\end{equation}
where $\bm{\epsilon}$ is the ground-truth noise to be subtracted from the image at time step $t$. After training, the model can generate new images by sampling from the noise distribution $\bm{z}_T \sim \mathcal{N}(0, \bm{I})$ and iteratively applying the reverse denoising process in \eqref{eq:reverse_denoising} for $T$ steps or performing jump sampling using the DDIM \cite{song2021denoising} to accelerate the generation without compromising quality, which can be expressed as \eqref{eq:ddim} in the next page.
\begin{figure*}[!ht]
    \begin{equation}
        \label{eq:ddim}
        \bm{z}_{t-1} = \sqrt{\alpha_{t-1}} \left( \frac{\bm{z}_t-\sqrt{1 - \alpha_t}\epsilon_\omega(\bm{z}_t, t)}{\sqrt{\alpha_t}} \right) + \sqrt{1 - \alpha_{t-1}} \epsilon_\omega(\bm{z}_t, t).
    \end{equation}
    \begin{equation}
    \label{eq:ddim_inversion}
\bm{z}_{t+1} = \frac{\sqrt{{\alpha}_{t+1}}}{\sqrt{\alpha}_{t}} \left( \bm{z}_{t} - \sqrt{1 - \alpha_{t}}\, \epsilon_\omega(\bm{z}_{t}, t, \bm{s})\right) + \sqrt{1 - \alpha_{t+1}}\, \epsilon_\omega(\bm{z}_{t}, t, \bm{s}).
\end{equation}
    \hrulefill 
\end{figure*}

    
\section{Proposed System}

\begin{figure}[!t]
    \centering
    \includegraphics[width=\linewidth]{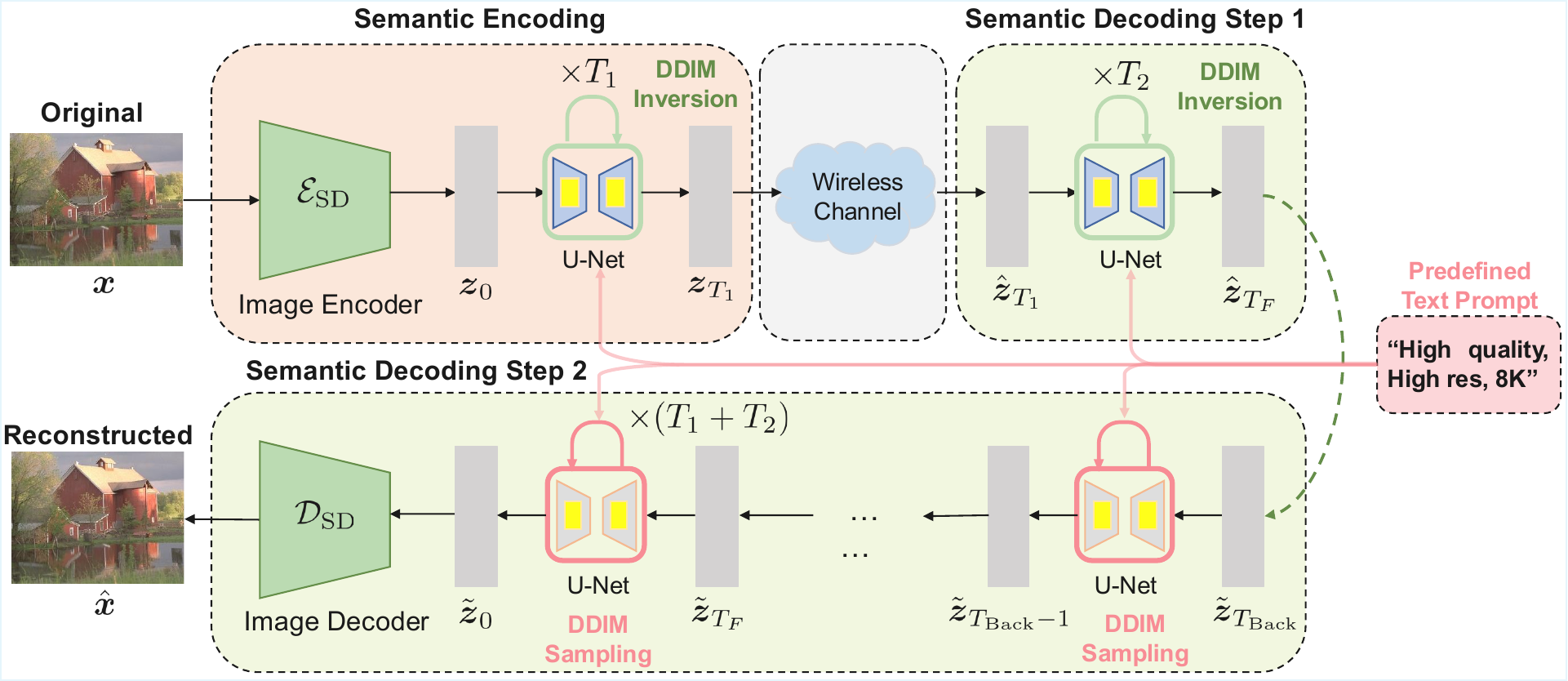}
    \caption{Overview of the proposed system, where the semantic encoder and decoder are all based on the pretrained stable diffusion model.}
    \label{fig:overview}
\end{figure}
Following the stable diffusion model, our proposed generative SemCom system comprises a CLIP text encoder $\text{CLIP}(\cdot)$, an image encoder $\mathcal{E}_{\text{SD}}(\cdot)$, a conditional U-Net model $\epsilon_\omega(\cdot,\cdot, \cdot)$, and an image decoder $\mathcal{D}_{\text{SD}(\cdot)}$, as \autoref{fig:overview} illustrates. At the transmitter side, the semantic encoder consists of a CLIP text encoder, an image encoder, and a U-Net model, which is used to extract the semantic features from the input image and generate the channel input signal. At the receiver side, there is also the same CLIP text encoder, U-Net model as the transmitter side, and an image decoder, which is used to reconstruct the image from the received signal. We note that all the components are pre-trained and fixed during the transmission process to align with the training-free design.
\begin{algorithm}[t]
    \caption{Procedure of the proposed training-free GenSemCom system}
    \label{alg:proposed}
    \begin{algorithmic}[1]
    \REQUIRE Input image $\bm{x}$, predefined text prompt $\bm{c}$, transmitter-side forward steps $T_{F,1}$, receiver-side forward steps $T_{F,2}$, denoising steps $T_{B}$ and text embedding $\bm{s}$
    \ENSURE Reconstructed image $\hat{\bm{x}}$
    \STATE \textbf{Transmitter:}
    \STATE Extract latent feature: $\bm{z}_0 = \mathcal{E}_{\mathrm{SD}}(\bm{x})$
    \STATE Perform DDIM inversion for $T_{F,1}$ steps conditioning on $\bm{s}$ to obtain $\bm{z}_{T_{F,1}}$ using Eq.~\eqref{eq:ddim_inversion}
    \STATE Transmit $\bm{z}=\gamma \bm{z}_{T_{F,1}}$ over the noisy channel
    

    \STATE \textbf{Receiver:}
    \STATE Set $\hat{\bm{z}}_{T_{F,1}} = \bm{y}$
    \STATE Perform DDIM forward process for $T_{F,2}$ steps to obtain $\hat{\bm{z}}_{T_{F}}$ where $T_{F} =T_{F,1} +T_{F,2}$ using Eq.~\eqref{eq:ddim_inversion}
    \STATE Initialize $\tilde{\bm{z}}_{T_{B}} = \hat{\bm{z}}_{T_{F}}$
    \STATE Perform DDIM sampling for $T_{B}$ steps  conditioning on $\bm{s}$ to obtain $\tilde{\bm{z}}_0$ using Eq.~\eqref{eq:ddim}
    \STATE Decode the reconstructed image: $\hat{\bm{x}} = \mathcal{D}_{\mathrm{SD}}(\tilde{\bm{z}}_0)$
    \end{algorithmic}
\end{algorithm}
\begin{figure*}[!ht]
    \centering
    \includegraphics[width=0.9\linewidth]{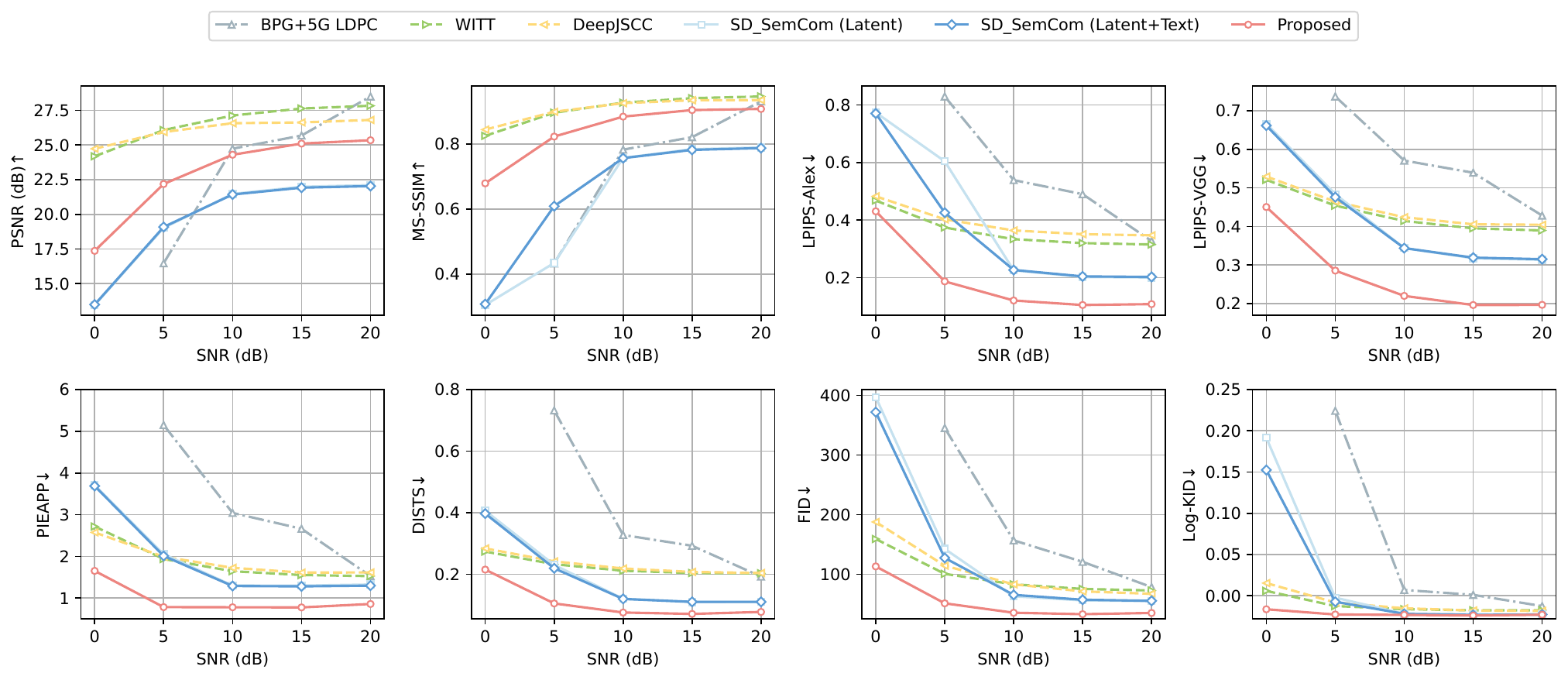}
    \caption{Reconstructed performance comparison for different methods, where BCR is set to 1/96 and SNR varies from 0 to 20dB}
    \label{fig:main_performance}
\end{figure*}
\subsection{Semantic Encoding via DDIM Inversion}
To extract the semantic information from the input image, we first use the image encoder to extract the latent feature, i.e., \ $\bm{z}_0 = \mathcal{E}_{\mathrm{SD}}(\bm{x})$. Unlike the work in \cite{Cicchetti_MLSP-2024} that directly transmits the latent feature $\bm{z}_0$ as well as the extracted caption, and then adds random noise in the diffusion forward process at the receiver side, we make three key modifications to improve the system performance as follows. 
\begin{enumerate}
    \item We propose to use the DDIM inversion to add $T_{F,1}$ steps of deterministic noise to the latent feature $\bm{z}_0$ before transmission. This process can be derived by reversing the DDIM sampling process in \eqref{eq:ddim}, and can be written as \eqref{eq:ddim_inversion} at the top of this page. The term $\bm{c}$ is a text prompt, $\bm{s} = \text{CLIP}(\bm{c})$ is the corresponding embedding produced by the CLIP text encoder, and $\epsilon_\omega(\bm{z}_{t}, t, \bm{s})$ is the predicted noise at time step $t$ by the U-Net model conditioned on $\bm{s}$. Thanks to the same U-Net model used in the forward and reverse processes, we can remove the noise more easily than the random noise added in \cite{Cicchetti_MLSP-2024}.

    \item The inversion process is performed at the transmitter side. This is because if we perform the DDIM inversion at the receiver side,  the noisy channel would distort the latent feature before inversion. This makes the UNet difficult to accurately estimate the noise at each forward timestep, deteriorating the invertibility of DDIM inversion.

    \item Instead of extracting and transmitting the image caption, we use a predefined text prompt such as ``\textit{High quality, High res, 8K}," which our experiments found to be more effective in improving reconstruction quality. \textcolor{black}{We note that $T_{F,1}$ forward diffusion steps are performed at the transmitter, and the obtained latent $\bm{z}_{T_{F,1}}$ is normalized by a scaling factor $\gamma = 1/\sqrt{\tfrac{1}{2k}||\bm{z}_{T_{F,1}}||_2^2}$ to satisfy the unit average power constraint\footnote{The factor $1/2$ comes from the conversion from real-valued to complex-valued signals.}. The normalized latent $\bm{z}=\gamma \bm{z}_{T_{F,1}}$is then mapped into complex-valued symbols for transmission.}
\end{enumerate}

\subsection{Semantic Decoding via DDIM Inversion and Sampling}
Upon receiving the noisy signal $\bm{y}$, the receiver first converts it into a real-valued vector. Next, the core idea of the proposed system is to set $\hat{\bm{z}}_{T_{F,1}} = \bm{y}$ and continue the DDIM forward process for $T_{F,2}$ steps, so that the noise level of this forward process aligns with the level of channel noise. We refer to the resulting noisy latent feature as $\hat{\bm{z}}_{T_{F}}$, where $T_{F} =T_{F,1} +T_{F,2}$. The rationale behind this forward process splitting is to avoid channel noise prediction at the transmitter by performing part of the forward process after channel perturbation at the receiver. Given this split architecture, we will provide a design guideline for the forward and backward processes in the next section. Note that when the receiver's computational capability is severely limited, a non-split architecture may be preferable. Optimizing such split model partitioning is an interesting direction for future study.


Next, the receiver performs the DDIM sampling process in \eqref{eq:ddim} to remove the $T_{F}$ steps of noise added to the latent feature. However, we note that due to the presence of channel noise, the actual noise level in $\bm{z}_{T_{F}}$ is no longer exactly equivalent to that introduced by $T_{F}$ DDIM steps. To address this issue, we propose to define $T_{B}>T_{F}$ as the number of steps to remove the noise added to the latent feature, and set $\tilde{\bm{z}}_{T_{B}} = \bm{z}_{T_{F}}$, and then perform the DDIM sampling process for $T_{B}$ steps to derive the denoised latent feature $\tilde{\bm{z}}_0$. Finally, the image decoder is used to reconstruct the image from the latent feature, i.e.\ $\hat{\bm{x}} = \mathcal{D}_{\mathrm{SD}}(\tilde{\bm{z}}_0)$. We summarize the pipeline of the proposed system in \autoref{alg:proposed}.

\subsection{Analysis}
We also provide a brief analysis of the proposed system from the perspective of matching the actual data distribution with that under the ideal training stage. Specifically, after applying $T_{F,1}$ steps of the forward diffusion process at the transmitter, introducing channel noise, and further applying $T_{F,2}$ steps of the forward process at the receiver, we characterize the resulting latent distribution in \autoref{prop:1}.

\begin{proposition}
    \label{prop:1}
 \textcolor{black}{Consider a latent feature $\bm{z}_0$ undergoing $T_{F,1}$ steps of the forward diffusion process, followed by normalization with a scaling factor $\gamma$. After being corrupted by the AWGN with variance $\sigma_{\rm ch}^2$, the latent further undergoes $T_{F,2}$ forward diffusion steps. The resulting latent $\hat{\bm{z}}_{T_F}$ follows}
\begin{equation}
    \hat{\bm {z}}_{T_F}\;\sim\;\mathcal N\Bigl(\gamma \sqrt{\bar\alpha_{T_F}}\;\bm z_0,\;(\sigma_\epsilon^2+\sigma_n^2)\,\bm I\Bigr),
\end{equation}
where
\begin{equation}
    \begin{aligned}
        \sigma_\epsilon^2  & = 1-\frac{\bar\alpha_{T_F}}{\bar\alpha_{T_{F,1}}}\,(1-\gamma^2)
        -\gamma^2\,\bar\alpha_{T_F}, \\
\sigma_n^2
&=\frac{\bar\alpha_{T_F}}{\bar\alpha_{T_{F,1}}}\,\sigma_{\rm ch}^2=\textcolor{black}{\Big(\textstyle\prod_{i=T_{F,1}+1}^{T_{F,1}+T_{F,2}} \alpha_i\Big) \sigma_{\rm ch}^2.}
\end{aligned}
\end{equation}
\end{proposition}

\begin{proof}[Proof Sketch]
    According to \eqref{eq:ddpm_forward}, after $T_{F,1}$ steps of the forward diffusion process, the latent at the transmitter can be written as
    \begin{equation}
        \bm{z}_{T_{F,1}} = \sqrt{\overline{\alpha}_{T_{F,1}}}\,\bm{z}_0 + \sqrt{1-\overline{\alpha}_{T_{F,1}}}\,\bm{\epsilon}.
    \end{equation}
    After transmission through the AWGN channel and normalization, the received latent becomes
    \begin{equation}
        \bm{y} = \gamma \bm{z}_{T_{F,1}} + \bm{n}.
    \end{equation}
    At the receiver, additional forward steps are recursively applied. For instance, after one step,
    \begin{equation}
        \hat{\bm{z}}_{T_{F,1}+1} = \sqrt{\alpha_{T_{F,1}+1}}\,\bm{y} + \sqrt{1-\alpha_{T_{F,1}+1}}\,\bm{\epsilon}.
    \end{equation}
    Further expanding the recursion, $\hat{\bm{z}}_{T_{F,1}+2}$ can be expressed as 
    \begin{equation}
    \begin{aligned}
        \hat{\bm{z}}_{T_{F,1}+2} &= \sqrt{\alpha_{T_{F,1}+2}}\,\hat{\bm{z}}_{T_{F,1}+1} + \sqrt{1-\alpha_{T_{F,1}+2}}\,\bm{\epsilon} 
        \end{aligned}
        \label{eq:z_T1+2}
        \end{equation}
        By continuing the recursion until $T_F$ steps, and using the reparameterization trick to combine the accumulated noise terms, we then prove \autoref{prop:1}.
    \end{proof}

Then, compared with the latent distribution under the ideal training stage in \eqref{eq:ddpm_forward},  we can obtain the key insights according to \autoref{prop:1}, as follows:
\begin{remark}[Necessity of applying forward diffusion at the transmitter]
According to the forward diffusion process in \eqref{eq:ddpm_forward}, applying more forward steps $T_{F,1}$ at the transmitter makes the transmitted latent $\bm{z}_{T_{F,1}}$ increasingly resemble a standard Gaussian distribution. This leads the normalization factor $\gamma$ to approach unity. Moreover, the noise variance $\sigma_\epsilon^2$ is affected by the ratio $\frac{\bar\alpha_{T_F}}{\bar\alpha_{T_{F,1}}}$. \textcolor{black}{Since the noise schedule $\{\alpha_i\}$ satisfies $\alpha_i\in(0,1)$ and typically decreases with the step $i$, setting $T_F > T_{F,1}$ (i.e., $T_{F,2} > 0$) leads to the multiplication of additional $\alpha_i$ terms, thereby decreasing the ratio $\bar\alpha_{T_F}/\bar\alpha_{T_{F,1}}$ and reduces the gap between $\sigma_\epsilon^2$ and the ideal noise variance $1-\bar\alpha_{T_F}$ in \eqref{eq:ddpm_forward}.}

\end{remark}

\begin{figure*}[!t]
    \centering
    \begin{subfigure}{0.25\linewidth}
        \centering
        \small Original
                \vspace{2pt}
        \includegraphics[width=\linewidth]{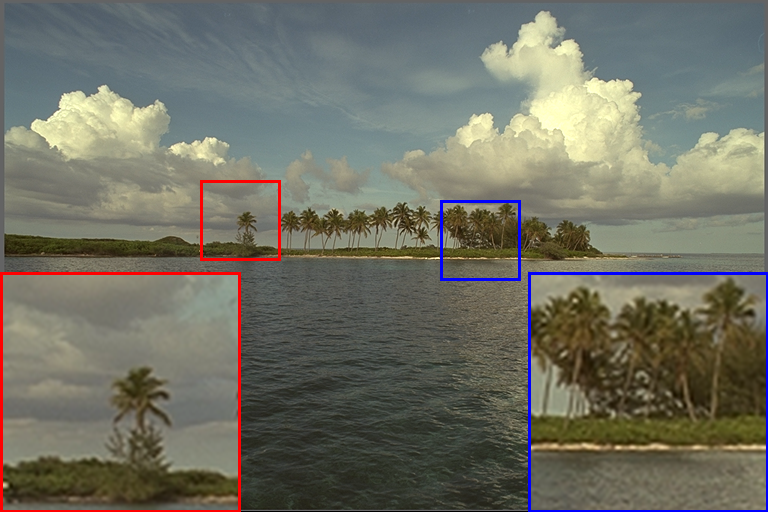}
        \small PSNR/LPIPS 
        \label{fig:original}
    \end{subfigure}
        \begin{subfigure}{0.25\linewidth}
        \centering
        \small BPG+LDPC
        \vspace{2pt}
        \includegraphics[width=\linewidth]{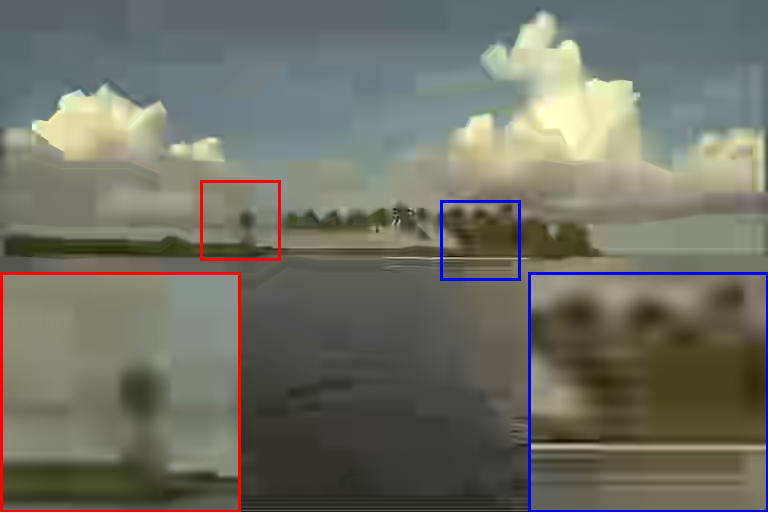}
        \small 26.50/0.4415
        \label{fig:bpgldpc}
    \end{subfigure}
    \begin{subfigure}{0.25\linewidth}
        \centering
        \small DeepJSCC
        \vspace{2pt}
        \includegraphics[width=\linewidth]{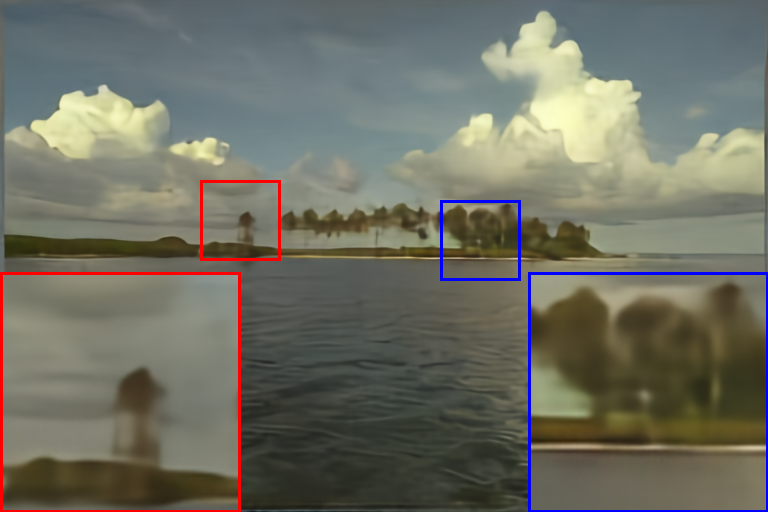}
        \small 28.16dB/0.4286
        \label{fig:deepjscc}
    \end{subfigure}
    \\        \vspace{3mm}
    \begin{subfigure}{0.25\linewidth}
        \centering
        \small WITT
        \vspace{2pt}
        \includegraphics[width=\linewidth]{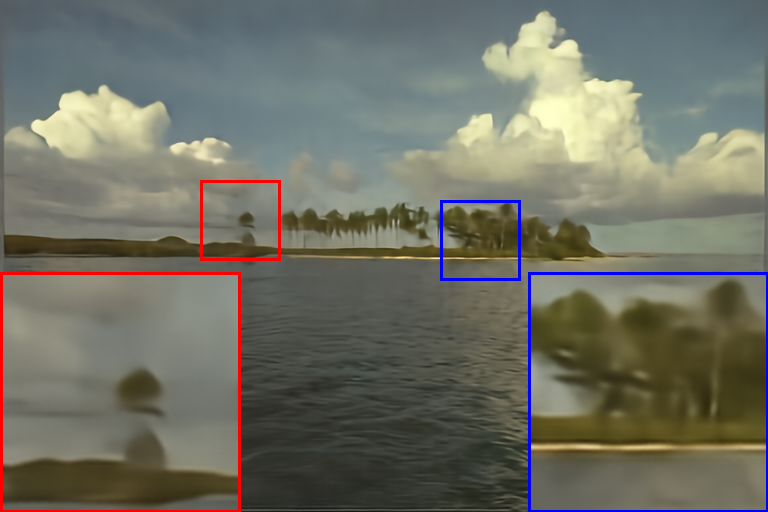}
        \small 28.23dB/0.4078
        \label{fig:witt}
    \end{subfigure}
    \begin{subfigure}{0.25\linewidth}
        \centering
        \small SD\_SemCom (Latent+Text)
                \vspace{2pt}
        \includegraphics[width=\linewidth]{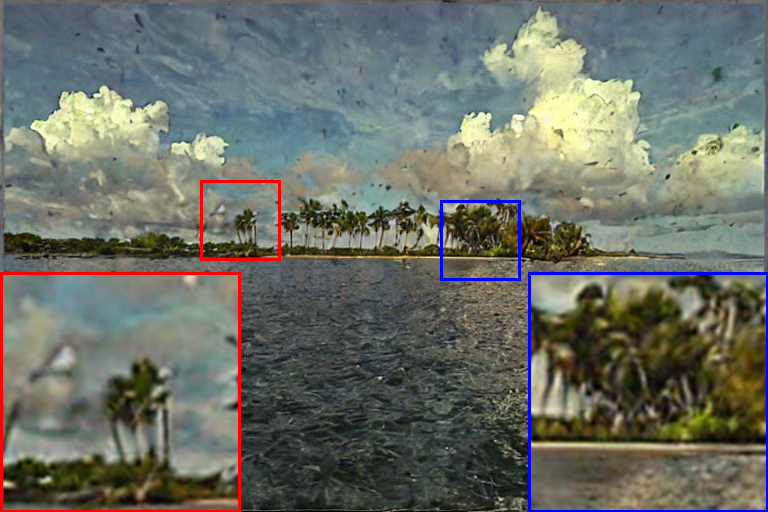}
        \small 20.67dB/0.4366
        \label{fig:sdsemcom}
    \end{subfigure}
    \begin{subfigure}{0.25\linewidth}
        \centering
        \small Proposed
                \vspace{2pt}
        \includegraphics[width=\linewidth]{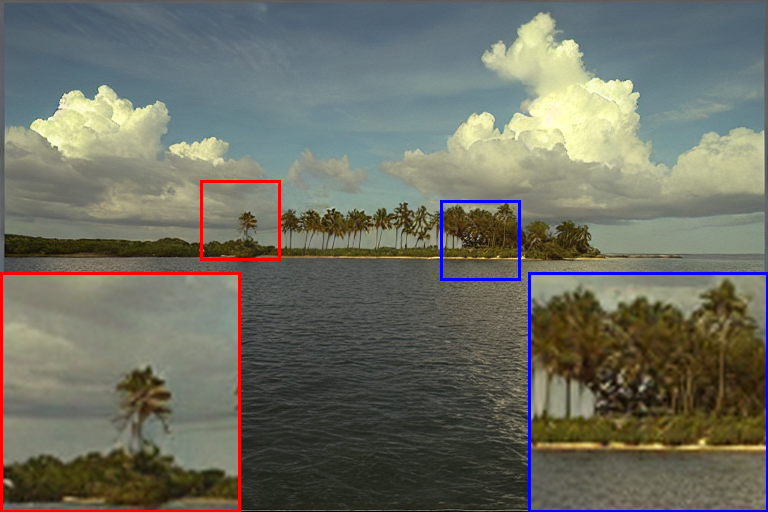}
        \small  24.67dB/0.2370
        \label{fig:proposed}
    \end{subfigure}
    \caption{Visual comparison at SNR of 5 dB. }
    \label{fig:vis_comp}
\end{figure*}

\begin{remark}[Necessity of continuing forward diffusion at the receiver]
    The contribution of the channel noise to the final latent is quantified by $\sigma_n^2 = \frac{\bar\alpha_{T_F}}{\bar\alpha_{T_{F,1}}}\,\sigma_{\rm ch}^2$. Similarly, by continuing the forward diffusion process for $T_{F,2}$ additional steps at the receiver, the ratio $\frac{\bar\alpha_{T_F}}{\bar\alpha_{T_{F,1}}}$ becomes smaller, which effectively reduces the impact of the channel noise.
\end{remark}

\begin{remark}[Necessity of performing additional denoising steps]
Due to the presence of channel noise, the total noise variance in the final latent $\hat{\bm{z}}_{T_F}$ becomes $\sigma_{\text{tot}}^2 = \sigma_\epsilon^2 + \sigma_n^2$, which is larger than the noise variance resulting purely from $T_F$ diffusion steps. In order to effectively remove the accumulated noise, the denoising process needs to be adjusted accordingly. Specifically, the equivalent number of denoising steps $T_{B}$ can be selected as
\begin{equation}
    1-\bar\alpha_{T_{B}} 	\approx \sigma_{\text{tot}}^2 \geq 1-\bar{\alpha}_{T_F},
\end{equation}
\textcolor{black}{where $1-\bar\alpha_{T_{B}}$ represents the noise level after $T_{B}$ steps forward without channel noise.} By setting $T_B > T_F$, the denoising schedule better matches the actual noise level, thereby improving the final reconstruction quality.
\end{remark}


\section{Simulations}\label{Sec:Simulation}

\subsection{Settings}
In simulations, we implement the semantic encoder and decoder using the pretrained stable diffusion 1.5 model\footnote{\url{https://huggingface.co/stable-diffusion-v1-5/stable-diffusion-v1-5}}, the BCR of the proposed system is set to $1/96$, and the SNR varies from 0 dB to 20 dB. During each transmission,
the text prompt is set to ``\textit{High quality, High res, 8K}" for all images with a guidance scale of 6. We empirically set the total steps of the noise scheduler of stable diffusion is $50$, and the number of forward steps $T_{F,1}$ and $T_{F,2}$ are set to 5 and 5, respectively. The number of denoising steps $T_{B}$ is set to decrease from $18$ to $10$ as the SNR increases from 0 dB to 20 dB. We evaluate the performance under the widely used Kodak dataset\footnote{\url{https://r0k.us/graphics/kodak/}} using various metrics, including PSNR, MS-SSIM, LPIPS, PIEAPP, DISTS, FID, and KID. We compare the proposed system with several baselines, including a traditional BPG compression followed by 5G LDPC for channel coding, an improved DeepJSCC system in \cite{GAN_JSCC}, WITT\cite{SWINJSCC}, and the training-free semantic communication system SD\_SemCom \cite{Cicchetti_MLSP-2024} with the same stable diffusion model.

\subsection{Effectiveness of the Proposed Method}
As shown in \autoref{fig:main_performance}, we compare the performance of the proposed system with the baselines, where SNR varies from 0 dB to 20 dB. From this figure, we can observe that the proposed system outperforms the baselines in most metrics. Specifically, as for the distortion metrics, the proposed system outperforms the competitive baselines with the same stable diffusion model, and also shows superior performance compared to the traditional digital communication system when the SNR is below 10 dB. In terms of the perceptual metrics and distribution metrics, the proposed system shows significant performance gains over all the baselines, which demonstrates the effectiveness of the proposed system.
\begin{table}[!t]
    \centering
    \caption{Comparison between random noise and DDIM inversion noise at SNR = 5 dB. 
    }
    \label{tab:noise_comparison}
    \begin{tabular}{lcccccc}
    \toprule
    Method  & PSNR ↑ & MS-SSIM ↑ & LPIPS ↓  & FID ↓ \\
    \midrule
    Random Noise & 21.72 & 0.767 & 0.211 & 58.75 \\
    Proposed (w/ Caption)  & \textbf{22.63} & \textbf{0.827} &\textbf{0.174} & \textbf{49.67}  \\
    Proposed (w/o Caption) & \underline{22.19} & 
    \underline{0.823} & \underline{0.187} & \underline{51.20}\\
    \bottomrule
    \end{tabular}
    \end{table}
    
In \autoref{fig:vis_comp}, we visualize the reconstructed images of the proposed system and the baselines, where the SNR is set to 5 dB. From this figure, we can find that the proposed system can reconstruct the image with better quality than the baselines. Specifically, the images reconstructed by BPG+LDPC, DeepJSCC, and WITT are extremely blurry and contain severe artifacts. Moreover, the image reconstructed by SD\_SemCom appears to be clearer with more realistic details. However, these images still suffer from severe artifacts and noise. In contrast, the image reconstructed by the proposed system shows significantly improved perceptual quality and fidelity, and is more similar to the original image, which further validates the effectiveness of the proposed system.
\addtolength{\topmargin}{0.05in}
\subsection{Ablation Studies}
We additionally conduct ablation studies to validate the effectiveness of the modifications made in the proposed system. As shown in \autoref{tab:noise_comparison}, we compare the performance of the proposed system with random noise and DDIM inversion noise, where the SNR is set to 5 dB. We also investigate the impact of not transmitting the caption. From this table, we can observe that the proposed system with DDIM inversion noise outperforms the one with random noise in all metrics, which demonstrates the effectiveness of the DDIM inversion used in the proposed system. Moreover, we find that the performance of the proposed system without transmitting the caption is only slightly lower than the version with a losslessly transmitted caption, with a gap of around 3\% in FID, which highlights the effectiveness of using a predefined prompt.

In \autoref{tab:ablation_steps}, we compare the performance of the proposed system with different settings of $T_{F,1}$, $T_{F,2}$, and $T_{B}$. From this table, we can observe that while the number of denoising steps $T_{B}$ is not enough, the performance of the proposed system is significantly degraded. This is because the noise level in the latent feature does not match the denoising process well. Moreover, while we set $T_{B}$ to be larger than $T_F=T_{F,1}+T_{F,2}$, the performance is improved. Besides, the configuration $T_{F,1}=5$ and $T_{F,2}=5$ outperforms the case with $T_{F,1}=10$ and $T_{F,2}=0$. These results align well with our analysis.

\begin{table}[!t]
    \centering
    \caption{Ablation study on the effect of $T_{F,1}$, $T_{F,2}$, and $T_{B}$ steps, whe SNR is set to 5 dB.}
    \label{tab:ablation_steps}
    \begin{tabular}{cccccccc}
    \toprule
    $T_{F,1}$ & $T_{F,2}$ & $T_{B}$ & PSNR ↑ & MS-SSIM ↑ & LPIPS ↓ & FID ↓ \\
    \midrule
    10 & 0 & 10 & 20.58 & 0.327 & 0.738 & 106.79 \\
    5 & 5 & 10 & 18.58 & 0.497 & 0.607 & 251.20 \\
    10 & 0 & 15 & \textbf{22.93} & \textbf{0.830} &\underline{0.192}   & \underline{54.04} \\
    5 & 5 & 15 & \underline{22.19} & \underline{0.823} & \textbf{0.187} & \textbf{51.20} \\
    0 & 10 & 15 & 21.83 & 0.793 & 0.197 & 53.07 \\
    \bottomrule
    \end{tabular}
\end{table}

\section{Conclusion}
In this paper, we have demonstrated that pretrained GDMs can serve as effective training-free JSCC for SemCom systems when appropriately adapted. By introducing a two-stage forward diffusion strategy and analyzing the impact of channel noise on the latent distribution, we have revealed how diffusion-based semantic encoding and denoising processes can align with noisy channel environments. Through extensive experiments, we have validated that the proposed system outperforms existing SemCom frameworks across diverse evaluation metrics.


\bibliographystyle{IEEEtran}
\bibliography{IEEEabrv, references}

@STRING{IEEE_J_JSAC       = "{IEEE} J. Sel. Areas Commun."}

@STRING{IEEE_J_CCN      = "{IEEE} Trans. Cogn. Commun. Netw."}

@STRING{IEEE_M_WC         = "{IEEE} Wireless Commun. Mag."}

@article{Semantic1,
  author       = {Deniz G{\"{u}}nd{\"{u}}z and
                  Zhijin Qin and
                  Inaki Estella Aguerri and
                  Harpreet S. Dhillon and
                  Zhaohui Yang and
                  Aylin Yener and
                  Kai{-}Kit Wong and
                  Chan{-}Byoung Chae},
  title        = {Beyond Transmitting Bits: Context, Semantics, and Task-Oriented Communications},
  journal      = {{IEEE} J. Sel. Areas Commun.},
  volume       = {41},
  number       = {1},
  pages        = {5--41},
  year         = {2023},
}

@article{DeepJSCC,
  author    = {Eirina Bourtsoulatze and
               David Burth Kurka and
               Deniz G{\"{u}}nd{\"{u}}z},
  title     = {Deep Joint Source-Channel Coding for Wireless Image Transmission},
  journal   = {{IEEE} Trans. Cogn. Commun. Netw.},
  volume    = {5},
  number    = {3},
  pages     = {567--579},
  year      = {2019},
}

@ARTICLE{GAN_JSCC,
  author={Erdemir, Ecenaz and Tung, Tze-Yang and Dragotti, Pier Luigi and Gündüz, Deniz},
  journal=IEEE_J_JSAC, 
  title={Generative Joint Source-Channel Coding for Semantic Image Transmission}, 
  year={2023},
  volume={41},
  number={8},
  pages={2645-2657},
  }

@ARTICLE{Shunpu_SemCom,
  author={Tang, Shunpu and Yang, Qianqian and Fan, Lisheng and Lei, Xianfu and Nallanathan, Arumugam and Karagiannidis, George K.},
  journal={IEEE Trans. Commun.}, 
  title={Contrastive Learning-Based Semantic Communications}, 
  year={2024},
  volume={72},
  number={10},
  pages={6328-6343},
}

@inproceedings{LPIPS,
  author       = {Richard Zhang and
                  Phillip Isola and
                  Alexei A. Efros and
                  Eli Shechtman and
                  Oliver Wang},
  title        = {The Unreasonable Effectiveness of Deep Features as a Perceptual Metric},
  booktitle    = {{IEEE/CVF} Int. Comput. Vis. Pattern Recognit. (CVPR)},
  pages        = {586--595},
  year         = {2018},
}

@INPROCEEDINGS{Cicchetti_MLSP-2024,
  title     = "Language-oriented semantic latent representation for image
               transmission",
  author    = "Cicchetti, Giordano and Grassucci, Eleonora and Park, Jihong and
               Choi, Jinho and Barbarossa, Sergio and Comminiello, Danilo",
  booktitle = {Proc. IEEE MLSP},
  pages     = "1--6",
  year      =  2024
}

@inproceedings{tang2024evolving,
  title={Evolving Semantic Communication with Generative Modelling},
  author={Tang, Shunpu and Yang, Qianqian and G{\"u}nd{\"u}z, Deniz and Zhang, Zhaoyang},
  booktitle={Proc. IEEE PIMRC},
  year={2024}
}

@article{SWINJSCC,
  author       = {Ke Yang and
                  Sixian Wang and
                  Jincheng Dai and
                  Xiaoqi Qin and
                  Kai Niu and
                  Ping Zhang},
  title        = {{SwinJSCC}: Taming Swin Transformer for Deep Joint Source-Channel Coding},
  journal      = {{IEEE} Trans. Cogn. Commun. Netw.},
  volume       = {11},
  number       = {1},
  pages        = {90--104},
  year         = {2025},
}

@article{AIGC,
  title={A comprehensive survey of {AI}-generated content ({AIGC}): A history of generative {AI} from {GAN} to {ChatGPT}},
  author={Cao, Yihan and Li, Siyu and Liu, Yixin and Yan, Zhiling and Dai, Yutong and Yu, Philip S and Sun, Lichao},
  journal={arXiv:2303.04226},
  year={2023}
}

@INPROCEEDINGS{Yilmaz-Other-2024,
  title     = "High perceptual quality wireless image delivery with denoising
               diffusion models",
  author    = "Yilmaz, Selim F and Niu, Xueyan and Bai, Bo and Han, Wei and
               Deng, Lei and Gündüz, Deniz",
  booktitle = {Proc. INFOCOM Workshp},
  volume    =  34,
  pages     = "1--5",
  year      =  2024
}

@ARTICLE{Chen-arXiveessIV-2025,
  title         = "{SING}: Semantic image communications using null-space and
                   {INN}-guided diffusion models",
  author        = "Chen, Jiakang and Yilmaz, Selim F and You, Di and Dragotti,
                   Pier Luigi and Gündüz, Deniz",
  journal       = {arXiv:2503.12484},
  year          =  2025,
}

@ARTICLE{Li_SemCom,
  author={Qiao, Li and Mashhadi, Mahdi Boloursaz and Gao, Zhen and Foh, Chuan Heng and Xiao, Pei and Bennis, Mehdi},
  journal={IEEE Wirel. Commun. Lett.}, 
  title={Latency-Aware Generative Semantic Communications With Pre-Trained Diffusion Models}, 
  year={2024},
  volume={13},
  number={10},
  pages={2652-2656}}

@INPROCEEDINGS{ho2020denoising,
  title={Denoising diffusion probabilistic models},
  author={Ho, Jonathan and Jain, Ajay and Abbeel, Pieter},
  booktitle={Proc. NeurIPS},
  volume={33},
  pages={6840--6851},
  year={2020}
}

@inproceedings{
song2021denoising,
title={Denoising Diffusion Implicit Models},
author={Jiaming Song and Chenlin Meng and Stefano Ermon},
booktitle={Proc. ICLR},
year={2021},

}

@ARTICLE{10614204,
  author={Liang, Chengsi and Du, Hongyang and Sun, Yao and Niyato, Dusit and Kang, Jiawen and Zhao, Dezong and Imran, Muhammad Ali},
  journal=IEEE_J_CCN, 
  title={Generative AI-Driven Semantic Communication Networks: Architecture, Technologies, and Applications}, 
  year={2025},
  volume={11},
  number={1},
  pages={27-47},
}

@article{tang2024retrieval,
  title={Retrieval-augmented generation for GenAI-enabled semantic communications},
  author={Tang, Shunpu and Zhang, Ruichen and Yan, Yuxuan and Yang, Qianqian and Niyato, Dusit and Wang, Xianbin and Mao, Shiwen},
  journal=IEEE_M_WC,
  year={2025}
}

\end{document}